\documentclass{appolb}
\usepackage{amsmath}
\usepackage{graphicx}
\usepackage{amssymb}
\usepackage{dsfont}

\begin{document}
\title{N-person quantum Russian roulette}
\author{Piotr Fr\c{a}ckiewicz$^1$ and Alexandre G M Schmidt$^2$
\address{$^1$ Institute of Mathematics of the Polish
Academy of Sciences\\ 00-956 Warsaw, Poland,\\ P.Frackiewicz@impan.gov.pl\\
$^2$ Departamento de F\'isica --- Universidade Federal
Fluminense\\ R. Des. Ellis Hermydio Figueira, 783, Volta Redonda,
RJ, Brazil, 27215-350,  agmschmidt@gmail.com }}\maketitle
\newtheorem{lemma}{Lemma}
\newtheorem{definition}{Definition}[section]
\newtheorem{fact}{Fact}[section]
\newtheorem{theorem}{Theorem}
\newtheorem{proposition}{Proposition}[definition]
\newtheorem{algorithm}{Algorithm}[section]
\newtheorem{example}[definition]{Example}
\newtheorem{corollary}[definition]{Corollary}
\newenvironment{proof}{\noindent\textit{Proof.}}
{\nolinebreak[4]\hfill$\square$\\\par}

\begin{abstract}
We generalize the concept of quantum Russian roulette introduced
in [A.G.M. Schmidt, L. da Silva, {\it Physica} A {\bf 392} (2013)
400-410]. Our model coincides with the previous one in the case of
the game with two players and gives the suitable quantum
description for any finite number of players.
\end{abstract}

\section{Introduction}

In order to model problems involving conflicts among individuals
von Neumann and Morgenstern \cite{vonneumann} introduced the
mathematical concept of games. These ideas found applications in
several fields of science, from Economics to Social Sciences and
reached the quantum domain back in 1999, when Meyer \cite{meyer}
and  Eisert, Wilkens and Lewenstein \cite{eisert} elaborated
quantum versions of a coin-flip game and of the prisoner dilemma
(PD), respectively. The classical PD game is the keystone for
modelling cooperative behaviour between animals, economic agents,
strategies for iterated games \cite{dyson} and even some RNA
virus\cite{nowak}. Quantum mechanics can solve such dilemma, as
were shown in reference \cite{eisert}, if one player could use a
quantum strategy --- quantum in the sense of an strategy that is
not cooperate nor defect. Experimentally the PD was realized using
NMR by Du and co-workers \cite{nmr} and recently using optical
techniques by Pinheiro {\it et al} \cite{pinheiro}. Several papers
followed these two seminal works, we can mention Monty-Hall
problem \cite{flitney}, discoordination games \cite{samaritan},
repeated quantum games \cite{piotr,donangelo}, multiplayer quantum
games \cite{multiplayer}, quantum auctions \cite{auction}, quantum
dating market \cite{dating} and minority game both theoretically
\cite{challet} and experimentally \cite{minority}.

On the other hand, game theory can model also not only cooperative
behaviour but conflict situations. This is the case of the
so-called quantum duel studied by Flitney and Abbott
\cite{flitney2} and revisited by Schmidt and Paiva \cite{schmid}.
In this game each player, Alice and Bob for instance, has a qubit
$|\psi\rangle = c_1|1\rangle + c_2|0\rangle$ where $|1\rangle$
represents the ``alive'' state and $|1\rangle$ the ``dead'' state,
and their only objective is to flip his/hers opponent's spin.
Following this idea Schmidt and da Silva proposed the quantum
version of the gamble known as Russian roulette, where players
shoot themselves at point blank range using a gun loaded with just
one bullet \cite{schmidt}. The authors found some interesting
results concerning the cases where the gun was fully loaded as
well as when there was just one quantum bullet inside its
chambers. However, in the case of 3-person game, one can find only
sketch of a possible form of the scheme rather than the exact
protocol. The framework, however, can lead to quite different
models, and some of them inconsistent with the 2-person quantum
roulette. The previous analysis suggests a possibility of studying
the game with larger number of players as the Authors construct
some particular case of quantum description of the 3-person.  Our
research is aimed at showing that there exists a natural extension
of the two-person quantum Russian roulette to the $n$-person case.

A formal theory for quantum games was initiated with the concepts
of quantum matrix games \cite{eisert, marinatto}. However, quantum
methods have found application in many other types of games and
decision problems from then on. Apart from the quantum scheme for
extensive games \cite{fracor3} one can find such quantized
problems as the quantum Monty Hall problem
\cite{lu,flitney,ariano} and the quantum roulette
\cite{wang,salimi}. The essential for our studying are ones
concerning quantum duels and truels \cite{flitney2,schmid}. It was
shown in paper \cite{schmidt} that the schemes for quantum duels
and truels can be adapted quantum Russian roulette where the
Authors provide us with the concept for two and three-person case.


The outline for our paper is the following: in
section~\ref{twopersonroulette} we review the two-person quantum
Russian roulette. Section~\ref{nroulette} is devoted to present
the generalized operators for the $n-$person case and to discuss
some illustrative examples. In the final section we conclude the
work.

\section{Two-person quantum Russian roulette}\label{twopersonroulette} According to
\cite{schmidt} the scheme for 2-person quantum Russian roulette is
described as follows: each player has a qubit and the game is
played in a Hilbert space $\mathds{C}^2 \otimes \mathds{C}^2$. The
game begins with both players alive, namely $|\psi_0\rangle =
|11\rangle$, where $i$-th qubit belongs to $i$-th player for
$i=1,2$. The referee prepares operators $O_{1}$ and $O_{2}$ (the
operators are also labelled $O_{i}(\gamma_{j})$ if the angle that
the operator depends on is relevant),
\begin{align}\label{operatoryschmidta}\begin{aligned}
O_{1} &= [e^{-i\alpha}\sin(\gamma/2)|11\rangle +
ie^{i\beta}\cos(\gamma/2)|01\rangle]\langle 11|\\
&\quad+[e^{i\alpha}\sin(\gamma/2)|01\rangle +
ie^{-i\beta}\cos(\gamma/2)|11\rangle ]\langle 01| +
|00\rangle \langle 00| + |10\rangle \langle 10|;\\
O_{2}&= [e^{-i\alpha}\sin(\gamma/2)|11\rangle +
ie^{i\beta}\cos(\gamma/2)|10\rangle]\langle 11|\\
&\quad+[e^{i\alpha}\sin(\gamma/2)|10\rangle +
ie^{-i\beta}\cos(\gamma/2)|11\rangle ]\langle 10| + |00\rangle
\langle 00| + |01\rangle \langle 01|.\end{aligned}
\end{align}
which the first and the second player act on the initial state.
The game proceeds in rounds (the referee is allowed to settle any
number of rounds). The state after $k$ rounds, $k=1,2,\dots,m$ is
defined recursively,
\begin{equation}\label{rekurencyjna}\begin{aligned}
&|\psi_{0}\rangle = |11\rangle,\\
&|\psi_{k}\rangle =
O_{2}(\gamma_{2k})O_{1}(\gamma_{2k-1})|\psi_{k-1}\rangle,
\end{aligned}
\end{equation}
and the outcome after $k-$th round is given by a measurement with
respect to the computational basis,
\begin{equation}\label{generalpayoff}
\langle S \rangle = \sum_{i,j=0,1}s_{ij}|\langle
ij|\psi_{k}\rangle|^2,
\end{equation}
where $s_{ij}$ is the outcome corresponding to the measured state
$|ij\rangle$. In particular, the measurement can be identified
with a payoff function. For example, in \cite{schmidt}, the
Authors studied the quantum Russian roulette with payoff functions
\begin{equation}\begin{aligned}
&\langle \mathdollar_{1} \rangle = \frac{1}{2}\left[1 + |\langle
10|\psi_{m}\rangle|^2 - |\langle 01|\psi_{m}\rangle|^2-|\langle
00|\psi_{m}\rangle|^2\right];\\
&\langle \mathdollar_{2} \rangle = \frac{1}{2}\left[1 + |\langle
01|\psi_{m}\rangle|^2 - |\langle 10|\psi_{m}\rangle|^2-|\langle
00|\psi_{m}\rangle|^2\right].
\end{aligned}
\end{equation}
which means that each of the players wins 1 if he is only one who
survived and 1/2 if both players are alive. Otherwise, they win
nothing. Note that, other payoff functions can be defined as well.
For example, the case that a player receives payoff 1 only when he
is the only survivor and the other player `receives' -1 is equally
natural, especially with reference to the classical game, where
the draw is not normally possible, i.e,
\begin{align}
\langle \mathdollar_{1} \rangle = |\langle 10|\psi_{m}\rangle|^2-
|\langle 01|\psi_{m}\rangle|^2, &\quad \langle \mathdollar_{2}
\rangle = |\langle 01|\psi_{m}\rangle|^2 - |\langle
10|\psi_{m}\rangle|^2.
\end{align}
Scheme (\ref{operatoryschmidta})-(\ref{generalpayoff}) generalizes
the classical 2-person Russian roulette with a gun having two
chambers in the barrel. Indeed, let us assume that $\gamma_{i} =
0$ represents a bullet in the player~$i$'s chamber and $\gamma_{i}
= \pi$ means that $i$-th chamber is empty. Then a result
corresponding to each of four possible cases in the classical
game, that is, when the barrel of the gun is full, a bullet is in
the first or the second chamber and the barrel is empty, can be
reconstructed with the results of scheme
(\ref{operatoryschmidta})-(\ref{generalpayoff}) for $(\gamma_{1},
\gamma_{2}) \in \{(0,0), (0,\pi), (\pi,0), (\pi, \pi)\}$. From
(\ref{rekurencyjna}) we have
\begin{align}\label{2resultingstates}\begin{aligned}
O_{2}(0)O_{1}(0)|11\rangle &= ie^{i\beta}|01\rangle, &
O_{2}(0)O_{1}(\pi)|11\rangle &= ie^{i(-\alpha + \beta)}|10\rangle;
\\  O_{2}(\pi)O_{1}(0)|11\rangle &=
ie^{i\beta}|01\rangle, & O_{2}(\pi)O_{1}(\pi)|11\rangle
&=e^{-2i\alpha}|11\rangle,\end{aligned}
\end{align} Thus, if the
states $|0\rangle$ and $|1\rangle$ are assumed to represent {\it
dead} and {\it alive} state, respectively, the measurement outcome
of each of the resulting states~(\ref{2resultingstates})
identifies with the classical results. For example, if the bullet
is inside the second player's chamber, the game ends with the
death of the second player. In the quantum case it means that
operators (\ref{operatoryschmidta}) are prepared with
$(\gamma_{1}, \gamma_{2}) = (\pi,0)$, and the measurement on
$O_{2}(0)O_{1}(\pi)|11\rangle$ outputs $|10\rangle$.

It worth noting that, in general, scheme
(\ref{operatoryschmidta})-(\ref{generalpayoff}) takes into
consideration the case in which a bullet in each chamber $i$ may
be found with some probability $p_{i}$. Indeed, the final state
associated with the general form of operators
(\ref{operatoryschmidta}) takes on the form
\begin{multline}\label{OOstate}
O_{2}(\gamma_{2})O_{1}(\gamma_{1})|11\rangle =
ie^{i\beta}\cos(\gamma_{1}/2)|01\rangle\\ + ie^{-i(\alpha
-\beta)}\sin(\gamma_{1}/2)\cos(\gamma_{2}/2)|10\rangle +
e^{-2i\alpha}\sin(\gamma_{1}/2)\sin(\gamma_{2}/2))|11\rangle
\end{multline}
The result given by the measurement (\ref{generalpayoff}) on state
(\ref{OOstate}) coincides with the classical result if we assume
$p_{i} \equiv \cos^2(\gamma_{i}/2)$. For example, the case where
the first player is alive and the second player fails at the end
of the first round takes place with probability $(1-p_{1})q_{1}$,
and the quantum counterpart $|10\rangle$ is measured with
probability $\sin^2(\gamma_{1}/2)\cos^2(\gamma_{2}/2)$.
\section{N-person quantum Russian roulette}\label{nroulette} Now, we extend
scheme (\ref{operatoryschmidta})-(\ref{generalpayoff}) to consider
the game with more than two players. We stick to the Russion
roulette in which the referee prepares the gun with one chamber
for each player, and a player finds a bullet in his own chamber
with some probability.

 The previous section shows that the proper definition of
2-person quantum Russian roulette lies in a construction of
suitable operators (\ref{operatoryschmidta}). Now, we use the same
reasoning for the $n$-person game.
 Let us denote by $\mathds{1}$ and $U_{i}$ the unitary operators on
$\mathbb{C}^2$ given by formulae
\begin{equation}\begin{aligned}
&\mathds{1} = |0\rangle \langle 0| + |1\rangle \langle 1|;\\
&U_{i} = \sin(\gamma_{i}/2)(
 e^{i\alpha}|0\rangle \langle 0| + e^{-i\alpha}|1\rangle \langle
 1|) + i\cos(\gamma_{i}/2)(e^{i\beta}|0\rangle \langle 1| +
e^{-i\beta}|1\rangle \langle 0|).
\end{aligned} \end{equation} Then operators (\ref{operatoryschmidta}) can be written as follows:
\begin{align}
O_{1} = \mathds{1}\otimes|0\rangle \langle 0| + U_{1} \otimes
|1\rangle \langle 1|, &\quad O_{2} = |0\rangle \langle 0| \otimes
\mathds{1} + |1\rangle \langle 1| \otimes U_{2},
\end{align}
The first term of each $O_{i}$ leaves the state unchanged if both
players are initially dead $(|00\rangle)$ or only player $i$ is
alive $(|10\rangle)$. Otherwise, the second term performs unitary
operation on player $i$'s qubit leaving the opponent's qubit
unchanged. It can be done the same in $n$-person case. That is,
player $i$'s operator $O_{i}$ does nothing if either only player
$i$ is alive or all the players are dead, and $O_{i}$ performs
unitary transformation only on player $i$'s qubit.
\begin{definition}\label{defnperson}
An $n$-person quantum Russian roulette with $m$ rounds is defined
by the following components:
\begin{itemize}
\item the players' operators
\begin{align}\label{trzyoperators}\begin{aligned}
&O_{1}=\mathds{1} \otimes (|0\rangle \langle 0|)^{\otimes n-1} +
\sum_{\substack{(i_{1},\dots,i_{n-1})\in \{0,1\}^{n-1}\\
(i_{1},\dots,i_{n-1}) \ne (0,\dots,0)}}U_{1} \otimes |i_{1} \dots
i_{n-1}\rangle \langle i_{1} \dots i_{n-1}|;\\
&O_{2} = |0\rangle \langle 0 | \otimes \mathds{1} \otimes
(|0\rangle \langle 0|)^{\otimes n-2} + \sum_{\substack{(i_{1},\dots,i_{n-1})\in \{0,1\}^{n-1}\\
(i_{1},\dots,i_{n-1}) \ne (0,\dots,0)}}|i_{1}\rangle \langle
i_{1}| \otimes U_{2} \otimes |i_{2} \dots i_{n-1}\rangle \langle
i_{2}
\dots i_{n-1}|;\\ & ~~ \vdots\\
&O_{n} = (|0\rangle \langle 0|)^{\otimes n-1} \otimes \mathds{1}+ \sum_{\substack{(i_{1},\dots,i_{n-1})\in \{0,1\}^{n-1}\\
(i_{1},\dots,i_{n-1}) \ne (0,\dots,0)}} |i_{1} \dots
i_{n-1}\rangle \langle i_{1} \dots i_{n-1}| \otimes U_{n}.
\end{aligned}\end{align}
\item the final state $|\psi_{k}\rangle$ at $k-$th round,
$k=1,\dots,m$,
\begin{align}\label{mstate}\begin{aligned}
&|\psi_{0}\rangle = |1\rangle^{\otimes n}\\
&|\psi_{k}\rangle = O_{n}(\gamma_{kn}) \cdots
O_{2}(\gamma_{(k-1)n+2})O_{1}(\gamma_{(k-1)n+1})|\psi_{k-1}\rangle
\end{aligned}
\end{align}
\item the measurement after $k$ round,
\begin{equation}\label{fmeasurement}
S = \sum_{(i_{1},\dots,i_{n}) \in
\{0,1\}^{n}}s_{i_{1},\dots,i_{n}}|\langle
i_{1},\dots,i_{n}|\psi_{k}\rangle|^2.
\end{equation}
\end{itemize}
\end{definition}
The clear way of describing (\ref{trzyoperators}) allows us to
easily prove that
\begin{proposition}
Operators defined by (\ref{trzyoperators}) are unitary.
\end{proposition}
\begin{proof}
Without loss of generality we prove that $O_{1}$ is a unitary on
$(\mathbb{C}^{2})^{\otimes n}$. Since $\mathds{1}$ and $U_{1}$ are
unitary operators we obtain
\begin{equation}
O_{1}^{\dag}O_{1} = O_{1}O_{1}^{\dag} = \mathds{1} \otimes
(|0\rangle \langle
0|)^{\otimes n-1} + \sum_{\substack{(i_{1},\dots,i_{n-1})\in \{0,1\}^{n-1}\\
(i_{1},\dots,i_{n-1}) \ne (0,\dots,0)}}\mathds{1} \otimes |i_{1}
\dots i_{n-1}\rangle \langle i_{1} \dots i_{n-1}|.
\end{equation}
But
\begin{equation}
\sum_{\substack{(i_{1},\dots,i_{n-1})\in \{0,1\}^{n-1}\\
(i_{1},\dots,i_{n-1}) \ne (0,\dots,0)}}\mathds{1} \otimes |i_{1}
\dots i_{n-1}\rangle \langle i_{1} \dots i_{n-1}| =
\mathds{1}^{\otimes n} - \mathds{1} \otimes (|0\rangle \langle
0|)^{\otimes n-1},
\end{equation}
which implies that $O_{1}^{\dag}O_{1} = O_{1}O_{1}^{\dag} =
\mathds{1}^{\otimes n}$.
\end{proof}
The scheme defined by (\ref{trzyoperators})-(\ref{fmeasurement})
generalizes a Russian roulette in a similar way as scheme
(\ref{operatoryschmidta})-(\ref{generalpayoff}) does.
\begin{example}
\textup{Let us consider a 3-person quantum Russian roulette with a
single round given by definition~\ref{defnperson}. Putting $n=3$
and $k=1$ into (\ref{trzyoperators}) and (\ref{mstate}) we obtain
the following final state:}
\begin{equation}\label{trzystate}
\begin{aligned} |\psi_{1}\rangle = O_{3}O_{2}O_{1}|111\rangle = &-e^{2\beta
i}\cos(\gamma_{1}/2)\cos(\gamma_{2}/2)|001\rangle\\ &-e^{(-\alpha
+
2\beta)i}\cos(\gamma_{1}/2)\sin(\gamma_{2}/2)\cos(\gamma_{3}/2)|010\rangle\\
&\quad + ie^{(-2\alpha +
\beta)i}\cos(\gamma_{1}/2)\sin(\gamma_{2}/2)\sin(\gamma_{3}/2)|011\rangle\\
&-e^{(-\alpha +
2\beta)i}\sin(\gamma_{1}/2)\cos(\gamma_{2}/2)\cos(\gamma_{3}/2)|100\rangle\\
&\quad + ie^{(-2\alpha +
\beta)i}\sin(\gamma_{1}/2)\cos(\gamma_{2}/2)\sin(\gamma_{3}/2)|101\rangle\\
&\quad + ie^{(-2\alpha +
\beta)i}\sin(\gamma_{1}/2)\sin(\gamma_{2}/2)\cos(\gamma_{3}/2)|110\rangle\\
&\quad+
e^{-3i\alpha}\sin(\gamma_{1}/2)\sin(\gamma_{2}/2)\sin(\gamma_{3}/2)|111\rangle.
\end{aligned}
\end{equation}
\textup{Bearing in mind that $\gamma_{i}$ equal to $0$ and $\pi$
are identified with an empty and full chamber of  player $i$,
respectively, and state $|0\rangle$ represents {\it dead} state
whereas $|1\rangle$ represents {\it alive} state, let us examine
the four possible deterministic cases:
\begin{itemize}
\item {\it No bullets} \quad An empty barrel in the gun makes all
the players alive in the classical game. The quantum counterpart
is accomplished then by taking $(\gamma_{1},
\gamma_{2},\gamma_{3}) = (\pi, \pi, \pi)$ in scheme
(\ref{trzyoperators})-(\ref{fmeasurement}) as the final state
$|\psi_{1}\rangle$ takes on the form
$O_{1}(\pi)O_{2}(\pi)O_{3}(\pi)|111\rangle =
e^{-3i\alpha}|111\rangle$. \item {\it One bullet} \quad The game
with exactly one bullet in chamber $i$ always ends with the dead
of player $i$ in the classical case. The same is obtained by means
of the quantum scheme. For example, if the second player has got a
bullet i.e. $(\gamma_{1}, \gamma_{2},\gamma_{3}) = (\pi, 0, \pi)$,
the state $|101\rangle$ since the final state is
$O_{1}(\pi)O_{2}(0)O_{3}(\pi)|111\rangle = ie^{(-2\alpha +
\beta)i}|101\rangle$. \item {\it Two bullets} \quad There are two
ways to play the 3-person Russian roulette, where at most one
chamber in the barrel is left empty. The game can be playing until
either one of the players shoots himself or only one player stays
alive. Definition~\ref{defnperson} concerns the latter case. For
example, if player 1's chamber is empty while the other players
have got bullets in their chambers, the first player is the only
survivor in the end of the classical game. The same occurs in the
quantum case since $(\gamma_{1}, \gamma_{2},\gamma_{3}) = (\pi,
0,0)$ implies the measurement result $|100\rangle$. \item {\it
Three bullets} \quad If $(\gamma_{1}, \gamma_{2},\gamma_{3}) = (0,
0,0)$, i.e., the gun is fully loaded, the scheme given by
definition~\ref{defnperson} also imitates the classical game as
the final state is $-e^{2\beta i}|001\rangle$. Since the third
player is the last to pull the trigger and his opponents' chambers
are full, player 3 becomes the only survivor just before his turn
to fire the gun, and therefore he wins the game.
\end{itemize}
Scheme (\ref{trzyoperators})-(\ref{fmeasurement}) generalizes the
Russian roulette game even if the referee prepares bullets in the
gun in a stochastic way, that is, if player $i$ finds a bullet in
his chamber with probability $p_{i}$. Similarly to the 2-person
case if it assumed that $p_{i}\equiv \cos(\gamma_{i}/2)$, the
results of the measurement on (\ref{trzystate}) are consistent
with the results of the classical game.}
\end{example}
Like in the case of quantum duels and truels, if $|\psi_{0}\rangle
= |1\rangle^{\otimes n}$, it is not possible to obtain superior
results in the quantum Russian roulette playing only one round.
However, if the game is played two or more rounds (without
measuring state $|\psi_{k}\rangle$ before the last round is
played), the results in the quantum game can differ significantly
from the corresponding results in the classical game.
\begin{example}\textup{
To reduce necessary calculations, let us make an assumption that
the players play two rounds and the referee prepares a gun in
which each player finds a bullet in their chamber with equal
probability. In other words, following formula (\ref{mstate}), for
$i=1,2,3$, player $i$ manipulates with $O_{i}(\gamma_{i})$ and
$O_{i}(\gamma_{i+3})$ such that $\gamma_{i},\gamma_{i+3} =
\pi/2$.}

\textup{Let us ask a question: what is the average probability
that all the players remain alive after two rounds. As probability
of measuring state $|111\rangle$ after two rounds is given by
$|\langle 111|\psi_{2}\rangle|^2$ in the quantum case, it is
sufficient to determine only the amplitude $a_{111}$ associated
with basis state $|111\rangle$ instead of the whole state
$|\psi_{2}\rangle$. The amplitude takes on the following form:}
\begin{equation}
a_{111} = \frac{e^{-i\alpha}}{2\sqrt{2}}+\frac{1}{8}\left(
e^{-2i\alpha} -3 e^{-4i\alpha} + e^{-6i\alpha}\right)
\end{equation}
\textup{Now, if we make a natural assumption that $\alpha$ and
$\beta$ are uniformly distributed, the expected probability
$E[|\langle 111|\psi_{2}\rangle|^2]$ is given by
\begin{equation}\label{result}
E[|\langle 111|\psi_{2}\rangle|^2] =
\frac{1}{4\pi^2}\int\limits_{0}^{2\pi}
\int\limits_{0}^{2\pi}|a_{111}|^2d\alpha d\beta = \frac{19}{64}.
\end{equation}
To compare result (\ref{result}) with the classical case, note
that the measurement on state $|\psi_{1}\rangle$ was not
performed. Thus, we should compare (\ref{result}) only with a
classical probability of the result of the second round, (given
that the result of the first round allows players to play one more
time), and such a probability is equal to 1/8. Thus, in the
quantum domain, the players all are more than twice as likely to
survive the second round. \\
In order to make another comparison of quantum Russian roulette
with the classical one, let us consider that not only $\alpha$ and
$\beta$ parameters are uniformly distributed, but all angles
$\gamma_i$ are also uniformly distributed. Now the probability for
our 3-person two rounds, quantum game is given by,
\begin{equation}
P_{3,2} = \frac{1}{4\pi^8} \int\limits_{0}^{2\pi}
\int\limits_{0}^{2\pi}d\alpha d\beta \int_0^\pi \int_0^\pi
\int_0^\pi \int_0^\pi \int_0^\pi \int_0^\pi |\langle
111|\psi_2\rangle|^2 d\gamma_1\cdots  d\gamma_6\label{int-gammas}.
\end{equation}
One can investigate similar problems when there are 4- and
5-persons in the gamble, which yield 8-fold and 10-fold angular
integrals analogous to (\ref{int-gammas}). Our results are
summarized in the table 1 and one can infer that the quantum
version produces greater probabilities for an ``all alive''
outcome.}\vspace{12pt}

\noindent {\bf Table 1} Probabilities for all players being alive
after the second round when parameters $\alpha$ and $\beta$ as
well as angles $\gamma_i$ are uniformly distributed. Quantum
version produces greater probabilities.
\\
\begin{center}
\begin{tabular}{|l|l|l|} \hline
Number of players&Quantum&Classical\\
\hline
3&$\frac{1}{8}+ \frac{1}{64}+\frac{7}{2\pi^4}$ & $\frac{1}{8}$\\
4&$\frac{1}{16}+\frac{1}{256}+\frac{6}{\pi^8}+\frac{23}{8\pi^4}$ & $\frac{1}{16}$ \\
5&$\frac{1}{32}+\frac{1}{1024}+\frac{15}{\pi^8}+\frac{77}{32\pi^4}$
& $\frac{1}{32}$\\ \hline
\end{tabular} \end{center}
\end{example}
\begin{example}\textup{
Another interesting example, concerning the 3-person quantum
Russian roulette, takes place when the referee is unfair and
prepares the gun with no bullets for two players, and the third
players' chambers are fully loaded. The first round obviously
reproduces the classical result since the player shoot himself and
ends in the dead state; the other two players remain in the alive
state. This dead state can be flipped back to $|1\rangle$ in the
second round --- in all three cases namely $\gamma_1=\gamma_4=0$
or $\gamma_2=\gamma_5=0$ or $\gamma_3=\gamma_6=0$ --- the final
state is proportional to $|111\rangle$, i.e., all player end alive
with probability equal to unity. The third round yields the same
classical result as well as the fourth round return to the second
one and so forth. A player with fully loaded chambers kills
himself after an odd number of rounds if final measurement takes
place, and remains alive if an even number of rounds were played.
}
\end{example}
\begin{figure}[t]
\centering
\includegraphics[scale=1]{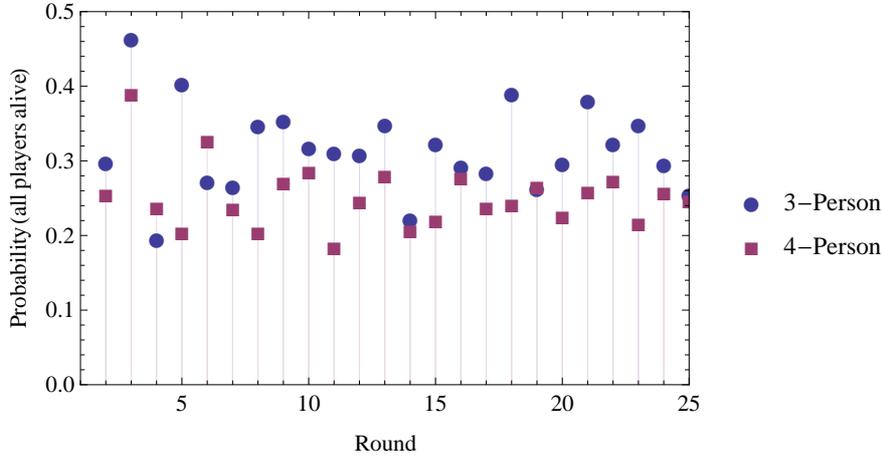}
\caption{\label{figure1}Plot of the probabilities $|\langle
a_{111}|\psi_{n}\rangle|^2$ and $|\langle
a_{1111}|\psi_{n}\rangle|^2$ as a function of the number of rounds
$n$ as studied in example 3.5.}
\end{figure}
\begin{figure}[here]
\centering
\includegraphics[scale=1]{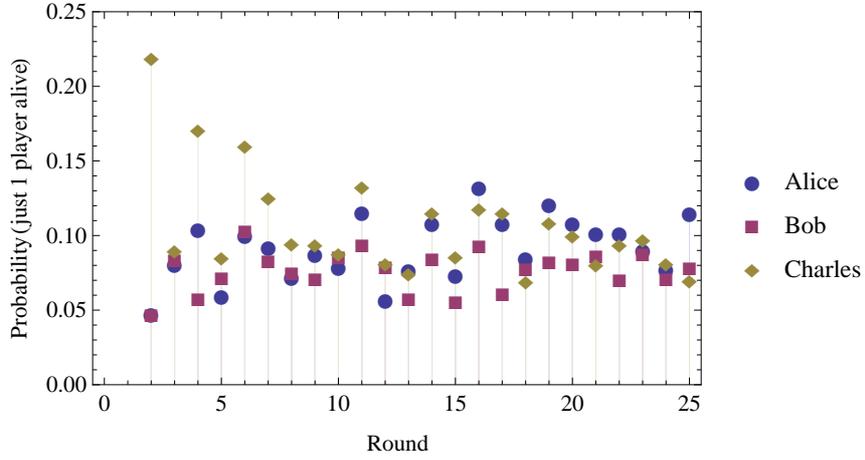}
\caption{Plot of the probabilities $|\langle
a_{100}|\psi_{n}\rangle|^2$ (Alice), $|\langle
a_{010}|\psi_{n}\rangle|^2$ (Bob) and $|\langle
a_{001}|\psi_{n}\rangle|^2$ (Charles) as a function of the number
of rounds $n$ as studied in example 3.5.} \label{figure2}
\end{figure}
\begin{example}
\textup{As a final example let us study what would happen if the
game was played several times, e.g., 25 rounds. Firstly let us
consider the case similar to example 3.3 where players have chance
of 50\% of having a bullet inside their chambers. One observes in
figure~\ref{figure1} that the maximum probability takes place in
the third round and does not exceed 48\%. On the other hand, we
infer from figure~\ref{figure2} that the last player to shot is
the one who is more likely to survive alone, namely his average
chance is $10.6\%$. Analogous conclusions can be inferred for the
4-person game where the maximum probability for an ``all alive''
outcome is around $40\%$ and occurs at the third round, see
figure~\ref{figure3}; as well as the last player has a slightly
greater chance of surviving alone, i.e., $6.2\%$. In both cases
the second player to shot is the one who has worst probabilities
of defeating
his opponents: $7.7\%$ and $3.9\%$ for 3-person and 4-person quantum Russian roulette respectively.\\
Conversely, if all the players have only one bullet inside their
chambers, and that each bullet is smeared out uniformly inside
each chamber, i.e., in a $n-$round game each player has
$\gamma_n=2\arccos{\frac{1}{\sqrt{n}}}$. Integrating over $\alpha$
and $\beta$ one observes that the probability of surviving
increases --- for 3- and 4-person the these probabilities exceed
80\% --- after each round since the players have smaller chance of
firing themselves, a problem analogous to the well-known
bomb-quest \cite{vaidman}.}
\begin{figure}[t]
\centering
\includegraphics[scale=1]{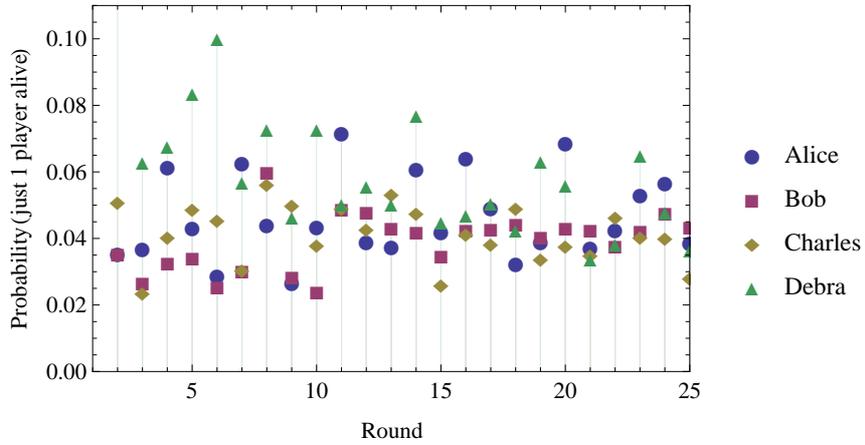}
\caption{Plot of the probabilities $|\langle
a_{1000}|\psi_{n}\rangle|^2$ (Alice), $|\langle
a_{0100}|\psi_{n}\rangle|^2$ (Bob), $|\langle
a_{0010}|\psi_{n}\rangle|^2$ (Charles) and $|\langle
a_{0001}|\psi_{n}\rangle|^2$ (Debra) as a function of the number
of rounds $n$ as studied in example 3.5.} \label{figure3}
\end{figure}
\end{example}

\section{Conclusion}
We generalized the two-person quantum Russian roulette game for an
arbitrary number of players. Our operators (\ref{trzyoperators})
allow for any preparation of the gun. We applied our results to
four  representative examples and compared the quantum game with
the classical one. In general the quantum game yields better
outcomes for the players.

\section*{Acknowledgments}
Work by Piotr Fr\c{a}ckiewicz was supported by the Polish National
Science Center under the project number DEC-2011/01/B/ST6/07197.
AGMS gratefully acknowledges FAPERJ (Funda\c c\~ao Carlos Chagas
Filho de Amparo \`a Pesquisa do Estado do Rio de Janeiro, grants
E26.110.192/2011 and E26.110.321/2012) and INCT-IQ for partial
financial support.

\end{document}